\documentclass{article}
\usepackage[utf8]{inputenc}
\providecommand{\keywords}[1]
{
  \small	
  \textbf{\textit{Keywords}} #1
}

\title{Qudit surface codes and hypermap codes}
\author{Zihan Lei}
\date{
School of Mathematical Sciences USTC,Hefei, China\\
lzihan9175@gmail.com\\
ORCID:0000-0002-7245-5958\\
\today}

\usepackage{amsmath}
\usepackage{amssymb}
\usepackage{amsthm}
\usepackage{appendix}
\usepackage{array}
\usepackage{caption}
\newtheorem{theorem}{Theorem}
\newtheorem{lemma}[theorem]{Lemma}
\usepackage{wrapfig}
\usepackage[export]{adjustbox}
\usepackage{tikz}
\usetikzlibrary{decorations.markings}
\usetikzlibrary{commutative-diagrams}
\DeclareMathOperator{\im}{im}

\DeclareMathOperator{\stAr}{star}
\DeclareMathOperator{\weight}{weight}
\usepackage{graphicx}
\graphicspath{ {images/} }

\begin{document}
\maketitle
\begin{abstract}
    In this article, we define homological quantum codes in arbitrary qudit dimensions $D\geq 2$ by directly defining CSS operators on a 2-Complex $\Sigma$. If the 2-Complex is constructed from a surface, we obtain a qudit surface code. We then prove that the dimension of the code we define always equals the size of the first homology group of $\Sigma$. We also define the distance of the codes in this setting, finding that they share similar properties with their qubit counterpart. Additionally, we generalize the hypermap-homology quantum code proposed by Martin Leslie to the qudit case. For every such hypermap code, we construct an abstract 2-Complex whose homological quantum code is equivalent to the hypermap code.
\end{abstract}

\keywords{qudit, 2-Complex, CSS code,  surface code, hypermap}

\section{Introduction}
Surface codes are an important class of error-correcting codes in fault-tolerant quantum computation. In the literature, rigorous constructions of surface codes are usually done for the case of $\mathbb{Z}_2$-vector spaces, which is reasonable because qubit quantum computation theories have been highly successful, and qubit quantum codes are still dominant in today's research. However, unlike qubits, which are two-level quantum systems, qudits possess a higher-dimensional state space. This inherent richness opens up new possibilities for information encoding and error correction \cite{10.1007/3-540-49208-9_27,PhysRevX.2.041021,PhysRevResearch.3.L042007,qudition}. By harnessing the power of qudits, we have the potential to design error correction codes that are more robust, efficient, and capable of correcting a greater number of errors than their binary counterparts. In the past decade,  some numerical studies have been conducted using  qudit  surface codes \cite{Sabo2021TrellisDF,Hussain,PhysRevA.92.032309} which is a special case of qudit topological quantum error correction codes. We believe that more study of similar kind are expected to be
conducted in these codes, such as the investigation of properties under coherent noise that has garnered attention in recent years \cite{coherent}. Therefore, understanding the general
structure of qudit topological quantum error correcting codes is helpful for research in this area. However, there seems to be limited literature specifically addressing such
a general structure. For this reason, this paper attempts to explore the construction
and basic properties of \(D \geq 2\) dimensional qudit topological quantum codes. Our results do not impose any specific requirements on D other than restricting it to integers lager than $1$, making the scope of applicability of our results sufficiently large. These results are very similar to the case of qubits,
although the proof methods may differ. However, things are not always the same.
For example, besides the necessary Theorem \ref{L trans}, we have not further investigated the
structure of logical operators as certain desirable properties of logical operators have more
requirements on \(D\), which we will explain further in the last part of section 3.
\par The general theory of qudit stabilizer and surface codes is introduced in Bombin and Martin-Delgado \cite{Bombin}, where the authors define them using symplectic codes. In this article, we provide a more direct construction of surface codes with arbitrary qudit dimension $D\geq 2$ in a way similar to the prevailing literature on qubit surface codes, such as Freedman and Meyer \cite{Projective}. Following Gheorghiu \cite{standard}, we define stabilizer codes simply as the subspace stabilized by a subgroup $\mathcal{S}$ of the qudit Pauli group. Then, we use the usual CSS construction to obtain $\mathcal{S}$ from an arbitrary 2-Complex defined in Bombin and Martin-Delgado \cite{Bombin}. When the 2-Complex is derived from a surface, we obtain the qudit surface code. Recall that in the case of a qubit stabilizer code, the number $k$ of logical qubits can be calculated by $k=n-r$, where $n$ is the number of physical qubits and $r$ is number of independent generators of $S$. This property can be generalized, which shows that  even in arbitrary qudit dimensions,   the “size” of a stabilizer code can be calculated by the size of its stabilizer group \cite{standard}. In this article, we use this generalization  (Theorem \ref{size of sta}) to relate the size of the homology group of a 2-Complex to that of its homological quantum code (Theorem \ref{size of homo}). This is more general than Theorem \uppercase\expandafter{\romannumeral3}.2 in Bombin and Martin-Delgado \cite{Bombin}, whose proof relies on the dimension theory of $\mathbb{Z}_2$-vector spaces. However,  we encounter $\mathbb{Z}_D$-modules in this article, which in general are not  vector spaces unless $\mathbb{Z}_D$ is a field, or equivalently, $D$ is a prime number. Next, we discuss the simplest kind of errors, i.e. Pauli errors, which permits us to study the distance of qudit homological quantum codes. In this part, we find that the main results are almost identical to the case of qubits \cite{bombin2013introduction, breuckmann2018phd}.
 \par As an example, we extend the hypermap-homology quantum code defined by Leslie \cite{Martin} to the case of qudits. Unlike Leslie's approach, which builds hypermap quantum codes from topological hypermaps, we define them directly from combinatorial hypermaps, making the statements more precise and rigorous at the cost of losing some geometric intuition. Then, for a given hypermap quantum code, we construct an abstract 2-Complex whose homological quantum code defined in this article is exactly equivalent to it. This approach was motivated by Sarvepalli's work, which shows that any (canonical) hypermap quantum code is equivalent to a surface code that can be directly constructed on its underlying surface.\cite{Pradeep}
 
\section{Preliminaries}
In this section, we review the necessary background material in a self-contained way. We will talk about qudit systems and their stabilizer codes in 2.1, then we will give the definition of 2-Complex in 2.2, which is helpful in clarifying the structure of qudit surface code and is also necessary in section 4.
\subsection{ Qudit Systems Of Dimension \(D^n\)}
A \emph{qudit} is a  finite dimensional quantum system  with  dimensioin $D\geq2$. As with the qubits' case, two operators
 $X$ and $Z$ act on a single qudit, and is defined as:
\begin{align}
  X & =\sum_{j\in \mathbb{Z}_D}{|j+1\rangle\langle{j}|} \\
    Z & =\sum_{j\in \mathbb{Z}_D}{\omega^j |j\rangle\langle{j}|}    
 \end{align}
 where $\omega=e^{2\pi{i}/D}$ and $\{|j\rangle\mid j\in {\mathbb{Z}_D}\}$ is an orthonormal basis for the qudit Hilbert space $\mathcal{H}$, also, the addition of integers in equation (1) should be treated to be done in $\mathbb{Z}_D$ . Notice that the definition of $X$ there is the adjoint $X^\dag$ of that in Gheorghiu.\cite{standard} From the above equations, we have \(ZX=\omega{XZ}\), and $X^D=Z^D=1$. As with the qubits' Hadamard gate, there are so called \emph{(discrete) Fourier gate} which maps the $\omega^k$-eigenvector \(|k\rangle\) of $Z$ to an $\omega^k$-eigenvector $|H_k\rangle$ of $X$, with
 \begin{equation}
     |H_k\rangle=\frac{1}{\sqrt{D}}\sum_{j}\omega^{-jk}|j\rangle.
 \end{equation}

For n qudits, the Hilbert space is denoted by $\mathcal{H}_n$ and we have
 \begin{equation}
     \mathcal{H}_n=\bigotimes_{i=1}^{n}\mathcal{H}
      \end{equation}
 with canonical basis the tensor products of \(|j\rangle\). Denote $X_i$ and $Z_i$ the corresponding $X,Z$ operators acting on the i-th qudit,  we call expressions of the form  
 \begin{equation}\label{Pauli operator}
     \omega^\lambda{X^\mathbf{x}Z^\mathbf{z}}=\omega^\lambda{X_1^{x_1}Z_1^{z_1}}\otimes{X_2^{x_2}Z_2^{z_2}}\otimes{\cdots}{\otimes}X_n^{x_n}Z_n^{z_n}
 \end{equation}
 the \emph{Pauli products},\cite{standard} where \(\lambda\) is an integer and the n-tuples $\mathbf{x}=(x_1,x_2,\cdots,x_n)$,  $\mathbf{z}=(z_1,z_2,\cdots,z_n)$ belongs to \(\mathbb{Z}_D^n\). These Pauli products is closed under multiplication and generate the \emph{Pauli group} \(\mathcal{P}_n\). 
 
 Similar to the case with qubits, we can examine a subgroup $\mathcal{S}$ of $\mathcal{P}_n$ and its stabilizer subspace $\mathcal{C}\mathrel{\mathop:}=\{\vert \phi \rangle\in\mathcal{H}_n \big| s\vert \phi \rangle=\vert \phi \rangle, \forall s\in \mathcal{S}\}$. Notice that the latter can also be defined for any subset of \(L(\mathcal{H}_n)\) ---- the vector space of all linear operators on $\mathcal{H}_n$.  Likewise, in order for $(\mathcal{C},\mathcal{S})$ to be called a stabilizer code, we require that $\mathcal{C}\neq{\{0\}}$, and a necessary condition for this is that $\mathcal{S}$ is commutative and does not contain any scalar multiplication $e^{i\theta}I\neq I$. Otherwise, we can always find a unit complex number $e^{i\theta}\neq 1$ such that $\forall |\phi\rangle\in \mathcal{C}, e^{i\theta}|\phi\rangle=|\phi\rangle$, which implies $\mathcal{C}={\{0\}}$. In fact, the condition of not containing scalar multiplications $e^{i\theta}I\neq I$ is both necessary and sufficient for $\mathcal{C}\neq{\{0\}}$, which is indicated by the following theorem \cite{standard}:
 \begin{theorem}\label{size of sta}
 let \(\mathcal{C}\) be the stabilizer subspace of the subgroup \(\mathcal{S}\subset\mathcal{P}_n\) which  does not contain any scalar multiplication other than identity.  Then
 \begin{equation}\label{numberofk}
     K\times|\mathcal{S}|=D^n,
     \end{equation}
     where K is the dimension of $\mathcal{C}$, \(|\mathcal{S}|\) is the size (cardinality)  of the stabilizer group $\mathcal{S}$.
 \end{theorem}
 
 It should be noted that the first sentence of Theorem 1 is slightly more stringent than the original one presented in Gheorghiu's work [3], as the proof actually relies on it. (See Appendix A.)  Furthermore, while the lacking of nontrivial scalar multiplications results in the commutativity of $\mathcal{S}$, the converse does not hold. 
 
 Additionally, we would like to mention a type of higher dimensional generalization for qubits, distinct from qudits, in which the dimension of the local quantum system is  $p^k$, where $p$ is a prime number \cite{1715533}. In this scenario, the operators $X$ and $Z$ are defined based on operations performed on finite fields $\mathbf{F}_{p^k}$, rather than  on $\mathbb{Z}_{p^k}$. It's important to note that these generalizations are not utilized further in the article.

\subsection{Oriented 2-Complex}
 In contrast to the case of qubits, the orientation of the underlying 2-Complex is relevant when constructing a qudit surface code. Therefore, we follow the definition of 2-Complexes as introduced by Bombin and Martin-Delgado.\cite{Bombin} An oriented  graph is a graph with each edge assigned an orientation. From a combinatorial perspective, an oriented graph consists of a set of vertices $V$, a set of (directed) edges $E$, and two incidence functions $I_s,I_t:E\rightarrow V$, which we refer to as the ``source" and ``target" functions. We say an edge $e$ goes or points from $I_s(e)$ to $I_t(e)$. Based on this graph, we can construct a set of ``inverse edges" $E^{-1}=\{e^{-1}\mid e\in{E}\}$, besides the fact that it originates from an oriented edge in the reverse direction, $e^{-1}$
is just a symbol here. We also define $(e^{-1})^{-1}\mathrel{\mathop:}={e}$ and $I_s(e^{-1})=I_t(e)$, $I_t(e^{-1})=I_s(e)$, which allows the inverse operation and the functions $I_s,I_t$ to be extended to the whole set $\bar{E}=E\cup{E^{-1}}$. 
Given an $n$-tuple of extended edges $(e_0,e_1,\cdots,e_{n-1})$ where $e_i\in \bar{E}$ with its index $i\in \mathbb{Z}_n$, and satisfies $I_t(e_i)=I_s(e_{i+1})$, then a \emph{closed walk of length $n$} is defined to be the equivalence class of these $n$-tuples under the equivalence relation of cyclic permutations, i.e., $(e_0,e_1,\cdots,e_{n-1})\sim(\tilde{e}_0,\tilde{e}_1,\cdots,\tilde{e}_{n-1})\Leftrightarrow e_{i+k}=\tilde{e}_i$ for some fixed $k\in \mathbb{Z}_n$, and for all $i$. We denote the equivalence class of $(e_0,e_1,\cdots,e_{n-1})$ by
\begin{equation}
w=[e_0,e_1,\cdots,e_{n-1}]
\end{equation}
which has a well-defined inverse
\begin{equation}
w^{-1}\mathrel{\mathop:}=[\epsilon_0,\epsilon_1,\cdots,\epsilon_{n-1}]
\end{equation}
with $\epsilon_i=e_{n-1-i}^{-1}$.
  
The 2-dimensional generalization of graphs is 2-Complexes. Formally speaking, an \emph{oriented 2-Complexes} is an oriented graph \(\Gamma=(V,E,I_s,I_t)\) with an additional set $F$ called ``faces"  plus a function $B:F\rightarrow{W_\Gamma}$ called gluing map, 
where $W_\Gamma$ is the set of all closed walks. Similarly, we add a set of symbols \(F^{-1} \mathrel{\mathop:}= \{f^{-1}\big|f \in F\}\), obtaining $\bar{F}=F\cup{F^{-1}}$. Now the domains of the inverse operation and gluing map $B$ can be expanded to $\bar{F}$ by defining $(f^{-1})^{-1} = f$ and:
\begin{equation}\label{gluing map}
    B(f^{-1})=B(f)^{-1}, \quad\forall{f\in \bar{F}}.
\end{equation}
Geometrically, the set $F$ and the function $B$ come from gluing faces in a 2 dimensional cell structure ( also called CW-Complex \cite{Groupoid} ), where a face $f$ in $F$ represents a 2-cell whose boundary is glued onto the underlying one-dimensional skeleton (which corresponds to a directed graph in combinatorics), and which is topologically equivalent to a closed disc. 
The 2-cell has an orientation determined by a normal vector field. According to a specific rule, say, the right-hand rule, this orientation induces a boundary orientation of the 2-cell, which naturally gives rise to a closed path in the directed graph. This corresponds to the map $B$. If we take the same 2-cell but choose the opposite normal vectors field, then we obtain $f^{-1}$, whose induced closed path is obviously opposite to the original one. This gives us equation (\ref{gluing map}). In this article, we forget about these geometric origins and will take the previous combinatorial definition, which does not allow singular cases like gluing 2-cells into a single point but is sufficient for our purpose, as Bombin and Martin-Delgado have pointed out in\cite{Bombin}.

Compact surfaces has cell structures and thus can be combinatorially represented by  2-Complexes. However, not every 2-Complex represent a surface, those that do come from surfaces must satisfy the conditions of \emph{Surface 2-Complex }. For the purpose of our discussion, we will not provide a formal definition of a surface 2-Complex. Interested readers can refer to Bombin and Martin-Delgado for more details. \cite{Bombin}

\section{Qudit homological quantum codes }
We begin our construction of homological quantum codes based on 2-Complexes in section 3.1, then in section 3.2, a  discussion is  given on their distance parameters. Logical operators and their current explorations are also briefly mentioned in section 3.3.  

\subsection{Construction of quantum codes}
  Given a 2-Complex $\Sigma = (V, E, I_s, I_t, F, B)$, we can
  define three $\mathbb{Z}_D$ modules $C_0(\Sigma)$, $C_1(\Sigma)$, and $C_2(\Sigma)$ as free modules generated by the sets $V$, $E$, and $F$, respectively.  For example, \(C_0(\Sigma)\) consists of all the formal sums \(r_1v_1+r_2v_3+\cdots+r_{|V|}v_{|V|}\), where $r_i\in \mathbb{Z}_D, v_i\in V$. Then a boundary operator $\partial_1: C_1(\Sigma)\rightarrow C_0(\Sigma)$ is defined to be the unique homomorphism such that \(\partial_1(e)=I_t(e)-I_s(e)\) for each $e\in E$. To define the boundary $\partial_2: C_2(\Sigma)\rightarrow C_1(\Sigma)$, first, for any closed walk \(w=[e_1^{\sigma_1},e_2^{\sigma_2},\cdots,e_h^{\sigma_h}],e_i\in E,\sigma_i=\pm 1\), we define 
  \(
      c_w\mathrel{\mathop:}=\sum_{i=1}^h\sigma_i e_i,
\),
  then $\partial_2$ is the unique homomorphism such that \(\partial_2(f)=c_{B(f)}\) for any $f\in F$. Now, there is a simple but important equation
  \begin{equation}\label{partial 2}
      \partial_1\circ \partial_2=0.
  \end{equation}
   Any two maps $\partial_i,i=1,2$ between 3 modules $C_i, i=0,1,2$ that satisfying equation (\ref{partial 2}) are said to form a \emph{chain complex} $C_2\stackrel{\partial_2}{\longrightarrow}C_1\stackrel{\partial_1}{\longrightarrow}C_0$, in other places, there are longer type of chain complexes, which will not be used in this article. We denote $Z_1(\Sigma)\mathrel{\mathop:}=\ker\partial_1$ whose elements are called \emph{cycles} and $B_1(\Sigma)\mathrel{\mathop:}=\im\partial_2$ whose elements are called \emph{boundaries}. Equation (\ref{partial 2}) tells us that $B_1(\Sigma)\subset Z_1(\Sigma) $, in particular, \(B_1(\Sigma)\) is a normal subgroup of $Z_1(\Sigma)$, and we have the \emph{first homology group} $H_1(\Sigma)$ as the quotient group
   \begin{equation}
      H_1(\Sigma)\mathrel{\mathop:}=Z_1(\Sigma)/B_1(\Sigma).
   \end{equation}
   
   By representing the matrix of $\partial_i$ in the natural bases 
$V,E$ and $F$, we can construct a specific type of stabilizer codes known as homological quantum codes \cite{Martin}. However, instead of relying on matrix arguments, we will introduce fundamental cohomology concepts \cite{Bombin, Martin}, which offer a more concise and geometrically insightful approach. First, some algebraic remarks. If $A$ is a module over a commutative ring $R$, then the set of all homomorphisms from $A$ to $R$ is an $R$-module called the \emph{dual modules} of $A$ and is denoted by \(A^*\mathrel{\mathop:}=\text{Hom}_R(A,R)\). Now if $F$ is a free $R$-module with a finite basis \(X\), for each $x\in X$, let \(x^*: F\rightarrow R\) be the homomorphism given by \(x^*(y)=\delta_{xy}\) ( \(\forall y\in X\) ), where \(\delta_{xy}\) equals \(0\) in \( R\) if \(x\neq{y}\), and \(1_R\) if \(x=y\). Then a basic fact is that \(F^*\) is a free $R$-module with basis $\{x^*\mid x\in X\}$. Denote $C^i(\Sigma)\mathrel{\mathop:}=C^*_i(\Sigma)$, and also $(c^i,c_i)\mathrel{\mathop:}=c^i(c_i)$ for any \(c^i\in C^i(\Sigma)\) and \(c_i\in C_i(\Sigma)\), the cobondary operator $\delta_{i+1}: C^i\rightarrow C^{i+1}, (i\in \{0,1\})$ is defined by
   \begin{equation}
       (\delta_i(c^{i-1}),c_i)\mathrel{\mathop:}=(c^{i-1},\partial_i(c_i)).\quad i=1,2.
   \end{equation}
   Then, by equation (\ref{partial 2}), we have the  \emph{cochain complex} $C^2(\Sigma)\stackrel{\delta_2}{\longleftarrow}C^1(\Sigma)\stackrel{\delta_1}{\longleftarrow}C^0(\Sigma)$ with
   \begin{equation}\label{delta 2}
        \delta_2\circ \delta_1=0.
   \end{equation}
   along with so called first \emph{coholmology group} \(H^1(\Sigma)\mathrel{\mathop:}=Z^1(\Sigma)/B^1(\Sigma)\) where the \emph{cocycles} is defined by $Z^1(\Sigma)\mathrel{\mathop:}=\ker \delta_2$, and \emph{coboundaries} by \(B^1(\Sigma)\mathrel{\mathop:}=\im \delta_1\). Now, let the star of a vertex \(v\in V\) to be the set \cite{Bombin} \(\stAr(v)\mathrel{\mathop:}= \{(e,\sigma)\in E\times\{1,-1\}\mid I_t(e^\sigma)=v\}\).
    Then we have a geometric explanation of $\delta_1$ by
   \begin{equation}\label{vvt}
       \delta_1(v^*)=\sum_{(e,\sigma)\in \stAr(v)}\sigma e^*,
   \end{equation}
   which can be shown by definition and indicates that to calculate $\delta_1(v^*)$,  we have to consider only the edges that incident on $v$. When an edge $e$ points at $v$, we add $e^*$, when it starts from $v$, we add $-e^*$, when it is a self-circle on $v$, we will add nothing into the sum.

   To construct a stabilizer code, we attach one qudit to each edge of the 2-Complex, 
   thus obtaining a $D^{|E|}$ dimensional Hilbert space $\mathcal{H}_{|E|}$. What we need  is to find an appropriate subgroup of $\mathcal{P}_{|E|}$. To begin, we define two sets of operators. 
   \begin{itemize}
       \item \emph{Face operators}: For each face f, we have \(\partial_2(f)=c_{B(f)}=\sum_{i=1}^h\sigma_i e_i\), where $\sigma\in\{1,-1\}$, an operator is defined by
       \begin{equation}\label{face O}
           B_f\mathrel{\mathop:}= \prod_{i=1}^h Z_i^{\sigma_i}
       \end{equation}
       with $Z_i$ the $Z$ operator on $e_i$'s qudit.
       \item \emph{Vertex operator}: For each vertex, we have equation (\ref{vvt}), an operator is defined by
       \begin{equation}\label{vertex O}
           A_v\mathrel{\mathop:}=\prod_{(e,\sigma)\in \stAr(v)}X_e^\sigma
       \end{equation}
       with $X_e$ the \(X\) operator on $e$'s qudit.
       
   \end{itemize}
   Notice that by letting some exponential $\sigma$ equal to $0$, we can rewrite product as  \(B_f= \bigotimes_{i=1}^{|E|} Z_i^{\sigma_i}\) and  \(A_v=\bigotimes_{i=1}^{|E|}X_i^{\sigma_i}\),  thus there is an $|E|$-tuple $\mathbf{v}_f=(\sigma_1,\sigma_2,\cdots,\sigma_{|E|})$ for each face operator, and an $|E|$-tuple $\mathbf{u}_v=(\sigma'_1,\sigma'_2,\cdots,\sigma'_{|E|})$ for each vertex operator. Unlike those in equation (\ref{face O}), in the tensor product form, there is possibility that for some $i$, $|\sigma_i|>1$, because the closed walk may intersect with itself, in for example the case of a non-orientable surface.  Now, the multiplication of two  operators of the same type corresponds to the addition of their \(|E|\)-tuples in \(\mathbb{Z}_{D}^{|E|}\), which indicates that the subgroup of \(\mathcal{P}_{|E|}\) generated by all the operators of the same type corresponds to a submodule  of the free module $\mathbb{Z}_D^{|E|}$ . We use $\mathcal{A}$ to represent the subgroup generated by all vertex operators and $r(\mathcal{A})$ for it's corresponding submodule, similarly, we  use $\mathcal{B}$ and $r(\mathcal{B})$ when it comes to face operators. Now, manifestly, \(r(\mathcal{B})\) ($r(\mathcal{A})$) is simply the set of coordinates of elements in $\im\partial_2$ ($\im\delta_1$) under the basis \(\{e\mid e\in E\}\)(\(\{e^*\mid e\in E\}\)). 
   \begin{lemma}\label{commu}
   The elements of \(\mathcal{B}\) commute with elements of \(\mathcal{A}\).
   \end{lemma}
   \begin{proof}
   For any $f\in F$ and $v\in V$, we have \((\delta_1(v),\partial_2(f))=(v,\partial_1\circ\partial_2(f))=0\) by equations (\ref{partial 2}) and (\ref{delta 2}), which implies that the inner product \(\mathbf{v}_f\cdot\mathbf{u}_v=0\) in $\mathbb{Z}_D^{|E|}$. If we denote $g^+\mathrel{\mathop:}=\sum_{i\in{I^+}}\sigma_i\sigma'_i$ with $I^+\mathrel{\mathop:}=\{i\in\{1,2,\cdots,|E|\}\big| \sigma_i\sigma'_i>0\}$ ,  $g^-\mathrel{\mathop:}=\sum_{i\in{I^-}}|\sigma_i\sigma'_i|$ with $I^-\mathrel{\mathop:}=\{i\in\{1,2,\cdots,|E|\}\big| \sigma_i\sigma'_i<0\}$ , we have $g^+-g^-\equiv 0\mod D$. Now, from the basic relation \(ZX=\omega{XZ}\), we have \(Z^{-1}X^{-1}=\omega{X^{-1}Z^{-1}}\), \(Z^{-1}X=\omega^{-1}{XZ^{-1}}\) and \(ZX^{-1}=\omega^{-1}{X^{-1}Z}\), which shows that if we interchange $B_f$ and $A_v$, we will have $\omega^{g^+}$ and $(\omega^{-1})^{g^-}$ being generated, these gives a net phase of $1$, thus commuting.
   \end{proof}

   Let \(\mathcal{S}\) be the subgroup generated by all $B_f$ and $A_v$, then by Lemma 2 it is abelian (All face  operators commute with each other because they are all $Z$-type operators. Similarly, all vertex operators commute with each other.) so that any element $s$ of $\mathcal{S}$ can be written as
   \begin{equation} \label{b.a}
       s=b\cdot a
   \end{equation}
   with \(b\in\mathcal{B}\) and \(a\in\mathcal{A}\). Now, $b$ is generated by face operators and therefore must be a $Z$-type operator \(b=\bigotimes_{i=1}^{|E|} Z_i^{\sigma_i}\), so, when it is any phase factor of identity, it must be the identity itself, because we have $\sigma_i \equiv 0 \mod{D}$ for all $i$. 
 The same for \(a\). We conclude that $s$ cannot be any scalar multiplication other than identity. Then by Theorem 1, there is a stabilizer code $\mathcal{C}$ defined by $\mathcal{S}$  which has dimension $K=D^{|E|}/{|\mathcal{S}|}$ and is called surface code when the 2-Complex comes from a surface with or without boundary. In qubit's case, it can be further showed that the number of logical qubits contained in \(\mathcal{C}\) equals the dimension of the first homology group, i.e, $\dim H_1(\Sigma)$. However, the arguments using the dimension property of vector spaces cannot be applied in general $D$-qudit's case, for a $\mathbb{Z}_D$-module is not vector space when $D$ is not a prime. Fortunately, the next theorem shows that even for arbitrary $D$, the size of $H_1(\Sigma)$ still gives the value of \(K\).
   \begin{theorem}\label{size of homo}
   For any 2-Complex $\Sigma$, let $\mathcal{S}$ be the subgroup generated by all face and vertex operators defined by equation (\ref{face O}) and (\ref{vertex O}), then the dimension $K$ of its stabilizer code $\mathcal{C}$ equals the size of \(H_1(\Sigma)\), i.e, we have
   \begin{equation}
       K=|H_1(\Sigma)|.
   \end{equation}
   \end{theorem}
   \begin{proof}
   By equation (\ref{b.a}) and the property that operators of the form \(X^\mathbf{x}Z^\mathbf{z}\) (See equation (\ref{Pauli operator}).) is a basis of $L(\mathcal{H}_n)$, 
   we have \(|\mathcal{S}|=|\mathcal{B}||\mathcal{A}|=|r(\mathcal{B})||r(\mathcal{A})|\), so \(K=D^{|E|}/{|\mathcal{S}|}=|C_1(\Sigma)|/(|\im\partial_2||\im \delta_1|)\), thus we only need to prove \(|C_1(\Sigma)|/|\im \delta_1|=|\ker\partial_1|\), i.e, $|\ker\partial_1|\cdot |\im \delta_1|=D^{|E|}$. Notice that if \(x\in \ker\partial_1\), then for every \(\alpha\in\im\delta_1\), there is a \(\beta\in{C^0(\Sigma)}\) such  that $\alpha=\delta_1 \beta$, and we have $(\alpha,x)=(\beta,\partial_1x)=0$. On the other hand, if $y\in C_1(\Sigma)$ such that for all $\alpha\in\im\delta_1$, \((\alpha,y)=0\), then for all \(\beta\in C^0(\Sigma)\), we have $(\beta,\partial_1y)=(\delta_1\beta,y)=0$, which means $\partial_1y=0$, i.e, $y\in\ker\partial_1$. These together shows that the set of coordinates of the elements in $\ker\partial_1$ is the submodule \(r(\mathcal{A})^{\perp}\) of $\mathbb{Z}_D^{|E|}$. Now,
   by Theorem 3.2 in Zhao et al, \cite{Free} (See Appendix B.) we 
 have $|r(\mathcal{A})||r(\mathcal{A})^{\perp}|=D^{|E|}$, which proves our result.
   
   \end{proof}
   
   As an example, let us compute the dimension of the stabilizer subspace for a code on the projective plane $\mathbb{R}\mathbf{P}^2$ and the torus $\mathbb{T}^2$. For $\mathbb{R}\mathbf{P}^2$, a standard 2-Complex is composed of a point $v$, an edge $e$ with $I_s(e)=I_t(e)=v$, and a face $f$ with $B(f)=[e,e]$. Thus, $\partial_2(f)=e+e=2e$, and we have $\im \partial_2\simeq 2\mathbb{Z}_D$. Furthermore, since $\partial_1(e)=v-v=0$, we have $\ker \partial_1=C_1(\Sigma)\simeq \mathbb{Z}_D$. Therefore, $H_1(\Sigma)\simeq{\mathbb{Z}_D/{2\mathbb{Z}_D}}$, which has two elements when $D$ is even, yielding $K=2$, and has one element when $D$ is odd, yielding $K=1$. Thus, when $D$ is an even integer greater than 2, $K$ cannot always be written in the form of $D^k,k\in \mathbf{Z}$, and saying that "the code contains $\log_D K$ logical qudits" is directly meaningless, but may suggest more subtle representations of the logical subspace. (See for example \cite{vuillot2023homological,unknown}.) 
   For $\mathbb{T}^2$, a standard 2-Complex is composed of a point $v$, two edges ${e_1,e_2}$ with $I_s(e_i)=I_t(e_i)=v$, and a face $f$ with $B(f)=[e_1,e_2,e_1^{-1},e_2^{-1}]$. We have $\partial_1(e_1)=0$ and $\partial_2(f)=e_1+e_2-e_1-e_2=0$, which implies $\text{im},\partial_2=0$ and $\ker\partial_1=C_1(\Sigma)$. Therefore, $H_1(\Sigma)\simeq\mathbb{Z}_D^{2}$, and the code has dimension $D^2$, implying the presence of two $D$-dimensional logical qudits. Note that $H_1(\Sigma)$ only depends on the homotopy class of the surface, so these computations apply to other 2-Complex representations of the above surfaces.

   \subsection{The distance}

   We try to generate the distance of stabilizer codes to qudit's case now. Distance is a parameter that is closely related to the error correcting ability of a quantum code, therefore, we will  talk about Pauli errors and their syndromes at first. For any qudit stabilizer code \((\mathcal{C},\mathcal{S})\), suppose our logical state $|\phi\rangle \in \mathcal{C}$ suffers an Pauli error \(E \in \mathcal{P}_n\), we can define it's error syndromes as in the qubit's case. Let us fix a set of generators \(\langle s_1, s_2, \cdots, s_r\rangle\) for $\mathcal{S}$, which doesn't have to be minimal in any sense, for example, the face and vertex operators in a surface code. Then, for any integer  $1 \leq l \leq r$, the \emph{syndrome} of $E$ corresponds to $s_l$ are defined to be the integer \(0\leq\beta^E_{l}\leq D-1\) such that $Es_l = \omega^{\beta^E_{l}} s_lE$. We can obtain $\beta^E_{l}$ by measuring the observable $s_l$ for the error state \(E|\phi\rangle\), here, $s_l$ is not necessarily hermitian but we can do projective measurement whenever it is normal. 

   The next concept we need is weight. For any Pauli operator  \(P\in\mathcal{P}_n\),  we define it's \emph{support} as the set of physical qudits on which $P$ acts non-trivially, in the case of 2-Complexes, it also refers to the set of edges to which  these qudits are attached. One notable difference compared to the qubit case is that if \(D > 2\), the same support can correspond to multiple distinct Pauli operators of the same type, for example, $Z$-type operator, whereas in the qubit case, there is essentially only one $Z$-type Pauli operator corresponding to it. This is why in qubit surface codes, we usually do not distinguish a $Z$-type operator from its support, or an $X$-type operator from its support which is often represented in the dual complex (For a formal definition of dual complex in the oriented case, see \cite{Bombin}.). For another difference related to this, let us consider, for example, the Toric code shown in Figure \ref{Qudit toric code} with a chain like support for an X-type error with two ends, non-trivial syndromes can occur only in the face operators ( filled in green ) at the two ends in the qubit case, but in general, they can occur in any face operators that the support passes by, depending on the specific forms of the error.
   The \emph{weight} of $P$, being denoted by $\weight(P)$, is now the cardinality (the number of elements in our situation) of its support.

   \begin{figure}[ht]
    \centering
    \begin{tikzpicture}[scale=0.6,transform shape]
     \tikzset{middlearrow/.style={
        decoration={markings,
            mark= at position 0.6 with {\arrow[scale=1.5]{#1}} ,
        },
        postaction={decorate}
    }
}

 \filldraw[white] (-1.6,4) circle (1pt) node[anchor=south west]{$\color{blue!60}{Z^{-1}}$};
  \filldraw[white] (-0.1,3) circle (1pt) node[anchor= west]{$\color{blue!60}{Z}$};
  \filldraw[white] (-1.6,2) circle (1pt) node[anchor=north west ]{$\color{blue!60}{Z}$};
   \filldraw[white] (-1.9,3) circle (1pt) node[anchor= east]{$\color{blue!60}{Z^{-1}}$};
  
  \filldraw[white] (4.4,2) circle (1pt) node[anchor=south west]{$\color{blue!60}{Z^{-1}}$};
  \filldraw[white] (5.9,1) circle (1pt) node[anchor= west]{$\color{blue!60}{Z}$};
  \filldraw[white] (4.4,0) circle (1pt) node[anchor=north west ]{$\color{blue!60}{Z}$};
  \filldraw[white] (4.1,1) circle (1pt) node[anchor= east]{$\color{blue!60}{Z^{-1}}$};

     \foreach \i in {-1,...,4}
      \foreach \j in {-1,...,3}{
        \fill[black] (2*\i,2*\j) circle(2pt);
        \ifnum \i > -1
         \draw[middlearrow={latex}] (2*\i-2,2*\j) -- (2*\i,2*\j);
         \fi
          \ifnum \j < 3
          \draw[middlearrow={latex}] (2*\i,2*\j) -- (2*\i,2*\j+2);
          \fi
        };
        \foreach \i in {0,...,4}
      \foreach \j in {-1,...,2}{
      \draw[middlearrow={>},black, opacity=0.2](2*\i-1,2*\j+1.6) arc (90:380:0.6);
      }

\fill[green,opacity=0.2](-2,2) rectangle (0,4);
  \fill[green,opacity=0.2](4,0) rectangle (6,2);
  \fill[orange,opacity=0.1](0,2) rectangle (2,4);
  \fill[orange,opacity=0.1](2,2) rectangle (4,4);
  \fill[orange,opacity=0.1](4,2) rectangle (6,4);
  \draw[line width=1.5,line cap=round, red!40, opacity=0.7] (-1,3) -- (5,3) -- (5,1);
       
    \end{tikzpicture}
    \caption{A support of a $X$-type operator (in red) in a qudit Toric code  }
    \label{Qudit toric code}
\end{figure}
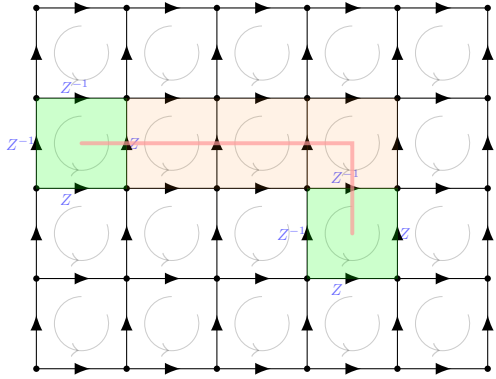

   We consider now the Pauli operators that leave the code space $\mathcal{C}$ invariant, these are the operators in $\mathcal{P}_n$ which commute with all elements in $\mathcal{S}$, i.e. the elements in the centraliser $C(\mathcal{S})$. Denote $N(\mathcal{S})$ as the normaliser of $\mathcal{S}$, we have 
   \begin{lemma}
       \(N(\mathcal{S}) = C(\mathcal{S})\).
   \end{lemma}
   \begin{proof}
       Consider $g\in N(\mathcal{S})$,  by the definition of normaliser, we have $\forall s \in \mathcal{S}$, there is an $s'\in \mathcal{S}$ such that \(gsg^\dag = s'\). Notice that there is also an integer \(\beta\) such that \(gsg^\dag=\omega^\beta sgg^\dag=\omega^\beta s\), we have \(s' = \omega^\beta s\) which means \(s'=s\), for otherwise we would have non-identity scalar multiplication in $S$. We proved \(N(\mathcal{S}) \subset C(\mathcal{S})\). The another side is trivial.
   \end{proof} 
\noindent As in the qubit case, by measuring the syndrome of $E$, we can cook up a possible correcting operator \(E'\) by  algorithms like Minimum weight matching decoder, hoping that 
\begin{equation}\label{requ}
    E'^\dag E|\phi\rangle = e^{i\theta} |\phi\rangle .
\end{equation}
Our algorithms only guarantee that $E'^\dag E\in N(\mathcal{S})$, which is true if and only if that the syndrome of \(E'\) equals the syndrome of $E$. However, when $E'^\dag E\notin \langle\omega I\rangle\mathcal{S}$, equation (\ref{requ}) may not be true for all encoded states \(|\phi\rangle\), because we have 
\begin{theorem}\label{L trans}
    Any element $R$ in $N(\mathcal{S})\setminus \langle\omega I\rangle\mathcal{S}$ will induce a linear transformation on $\mathcal{C}$ beyond scalar multiplications. 
\end{theorem}
\begin{proof}
    Suppose we have $R|\phi\rangle= e^{i\theta} |\phi\rangle$ for all $|\phi\rangle\in \mathcal{C}$, then $s'=e^{-{i\theta}}R$ should also stabilize $\mathcal{C}$ as well as the group $\mathcal{S}'$ which is generalised by $s'$ and all elements in $\mathcal{S}$. Thus, $\mathcal{S}'$ contains no non-identity scalar multiplication, otherwise its stabilizer subspace would be $\{0\}$. On the other hand, $s'$ must lies outside $S$, otherwise we would have \(R\in\langle\omega I\rangle\mathcal{S}\), for any two operators in \(\mathcal{P}_n\) that differ by a phase factor should differ by some $\omega^\beta I$. Now, by a modified version of Theorem \ref{size of sta} 
  (See the remark at the end of Appendix A. We resort to this stronger version because $S'$ may not be any subgroup of $\mathcal{P}_n$.), the stabilizer subspace \(\mathcal{C}'\) of $\mathcal{S}'$ is strictly smaller than \(\mathcal{C}\), which is a contradiction. 
\end{proof}
\noindent This indicates that \emph{logical errors} may still exist even after error correcting procedures. ( Theorem \ref{L trans}'s object of study is generally called \emph{undetectable errors} \cite{forlivesi2023performance} because they have zero syndromes and thus cannot be detected by syndrome measurements.) Now, we define the \emph{distance} $d$ of the code to be the lower bound of the weights of operators in $N(\mathcal{S})\setminus \langle\omega I\rangle\mathcal{S}$:
\begin{equation}
    d=\min_{R\in N(\mathcal{S})\setminus \langle\omega I\rangle\mathcal{S} }\weight(R). 
\end{equation}
\noindent If $\weight(E) < \frac{d}{2}$,  Minimum weight matching decoder will produce a corrector $E'$ such that $\weight(E')\leq \weight(E) < \frac{d}{2}$, then we have \(\weight(E'^\dag E)\leq \weight(E'^\dag)+ \weight(E) < d\), which means that $E'^\dag E\in \langle\omega I\rangle\mathcal{S}$ and thus the success of error correction. Intuitively speaking, Pauli errors with bigger weight are less likely to happen than those with smaller weight, if $d$ is big enough, errors with $\weight(E) < \frac{d}{2}$ will dominate, thus the logical error rate after error correction is small. Therefore, $d$ is a parameter that roughly describe the error correcting ability of the code. 

In the case of a homological quantum code \((\mathcal{S},\mathcal{C})\) from a 2-Complex $\Sigma$ described in section 3.1, $d$ can be given by pure geometric parameters of $\Sigma$. To begin with, we undertake some preparatory work. 
\begin{lemma}\label{Z,X}
    A Pauli operator \(P=\omega^\lambda{X^\mathbf{x}Z^\mathbf{z}}\) in equation (\ref{Pauli operator}) lies in \(N(\mathcal{S})\) if and only if $\mathbf{x}\in r(\mathcal{B})^\perp$ and $\mathbf{z}\in r(\mathcal{A})^\perp$. Moreover, $P$ lies in $\langle\omega I\rangle\mathcal{S}$ if and only if $\mathbf{x}\in r(\mathcal{A})$ and $\mathbf{z}\in r(\mathcal{B})$.
\end{lemma}
 \begin{proof}
     Recall the proof of  Lemma \ref{commu}, we conclude that $Z^\mathbf{z}X^\mathbf{x} = \omega^{\mathbf{x}\cdot\mathbf{z}} X^\mathbf{x}Z^\mathbf{z}$ holds for any $\mathbf{x},\mathbf{z}
     \in \mathbb{Z}^{|E|}_D$, which indicates that \(Z^\mathbf{z}\) and \(X^\mathbf{x}\) commute if and only if \(\mathbf{x}\cdot\mathbf{z}=0\) in \(\mathbb{Z}_D\). Following this property, the first part of the lemma is proved. The second part of the lemma is obvious.
 \end{proof}

 \noindent By lemma \ref{Z,X}, we have
 \begin{theorem}\label{weight CSS}
     The distance of the code can be expressed as
     \begin{equation}
         d = \min_{\mathbf{a}\in W}\weight(\mathbf{a}), 
     \end{equation}
     where \(\weight(\mathbf{a})\) which is called \emph{Hamming weight of $\mathbf{a}$} is the number of the nonzero components of the tuple $\mathbf{a}$, and $W=r(\mathcal{B})^\perp\setminus r(\mathcal{A})\cup r(\mathcal{A})^\perp\setminus r(\mathcal{B})$.
 \end{theorem}
 \begin{proof}
     we set \(d'=\min\{\weight(\mathbf{a})\big| \mathbf{a}\in W\}\). For any $R=\omega^\lambda{X^\mathbf{x}Z^\mathbf{z}}$  in $ N(\mathcal{S})\setminus \langle\omega I\rangle\mathcal{S}$, by Lemma \ref{Z,X}, at least one of the tuples $\mathbf{x}$ and $\mathbf{z}$ must lie in $W$, say $\mathbf{x}$, thus $\weight(R)\geq \weight(\mathbf{x})\geq d'$. Because $R$ is arbitrary, we have $d\geq d'$. On the other hand, $W$ is finite, there exist an \(\mathbf{a} \in W\) such that $\weight(\mathbf{a})=d'$, without loss of generality, we assume $\mathbf{a}\in r(\mathcal{B})^\perp\setminus r(\mathcal{A})$, then by Lemma \ref{Z,X}, we have $X^\mathbf{a}\in N(\mathcal{S})\setminus \langle\omega I\rangle\mathcal{S}$, thus $d'=\weight(X^\mathbf{a})\geq d$. We conclude that \(d'=d\).
 \end{proof}

 \noindent We refer to Theorem \ref{weight CSS} as a theorem because it applies to general CSS codes, not only homological quantum codes. For its qubit version and a systematic treatment of qubit CSS codes  in a similar style to this article, see Breuckmann\cite{breuckmann2018phd}. The next theorem whose qubit version is well-known is the core this subsection, it specifically targets homological quantum codes. To begin, recall that the set of coordinates of the elements in $\im \partial_2,\im\delta_1,\ker\partial_1$ under the basis of edges are the submodules \(r(\mathcal{B}),r(\mathcal{A}),r(\mathcal{A})^{\perp}\) of $\mathbb{Z}_D^{|E|}$. It is also not hard to show that $r(\mathcal{B})^\perp$ corresponds to $\ker\delta_2$ in this way. We list these correspondences in Table \ref{CBS}.

 \begin{table}[h]
 \captionsetup{skip=5pt}
 \caption{Correspondence between submodules}
 \label{CBS}
 \centering
 \renewcommand\arraystretch{1.3}
 
\begin{tabular}{ |c|c|c|c| @{\extracolsep{\fill}}|}

\hline
submodule of \(C_1{(\Sigma)}\)  & submodule of \(\mathbb{Z}^{|E|}_D\) \\
\hline
$\ker\partial_1$ & \(r(\mathcal{A})^\perp\) \\ 
\(\im\partial_2\) & \(r(\mathcal{B})\)  \\
\hline\hline
 submodule of \(C^1{(\Sigma)}\) &  submodule of \(\mathbb{Z}^{|E|}_D\) \\
 \hline
 \(\ker\delta_2\) & \(r(\mathcal{B})^\perp\) \\
 \(\im\delta_1\) & \(r(\mathcal{A})\) \\ 
\hline
\end{tabular}
\end{table}

\begin{theorem}
     For any homological quantum code in section 3.1, its distance is the smallest number chosen between the length \(l(c)\) of the shortest  nontrivial cycle $c$ and length \(l(c')\) of the shortest  nontrivial cocycle \(c'\), i.e, 
     \begin{equation}\label{d of homo}
          d=\min\{\min_{c\in \ker\partial_1\setminus \im\partial_{2}}l(c), \quad\min_{ c'\in \ker\delta_{2}\setminus \im\delta_{1}}l(c')\},
     \end{equation}
     where the length $l(c)$ is defined as the weight of coordinate $|E|$-tuple of $c$ in $\mathbb{Z}^{|E|}_D$, similar for $l(c')$, and the term `nontrivial' refers to fact that $c$ ($c'$) is not a boundary (coboundary).
\end{theorem}
\begin{proof}
    Again, we denote the right part of equation (\ref{d of homo}) as $d'$. Given $\mathbf{c}\in W$ ( $W$ is defined in Theorem \ref{weight CSS}.), suppose  $\mathbf{c}\in r(\mathcal{B})^\perp\setminus r(\mathcal{A})$, by Table \ref{CBS}, there exist a nontrivial cocycle $c'$ whose coordinate is $\mathbf{c}$, therefore $\weight(\mathbf{c})=l(c')\geq d'$, by Theorem \ref{weight CSS} and the arbitrariness of $\mathbf{c}$, we have $d\geq d'$. On the contrary, for any nontrivial cycle \(c\in \ker\partial_1\setminus \im\partial_{2}\), its coordinate $\mathbf{c}$ belongs to $r(\mathcal{A})^\perp\setminus r(\mathcal{B})$, we have \(l(c)=\weight(\mathbf{c})\geq d\), the same for any nontrivial cocycle, therefore $d'\geq d$. 
\end{proof}

\subsection{Discussion}
In section 3.2, We have shown that the actions of elements in $N(\mathcal{S})\setminus \langle\omega I\rangle\mathcal{S}$ on $\mathcal{C}$ are nontrivial, these elements are called \emph{logical operators}. Notice further that $\langle\omega I\rangle\mathcal{S}$ is a normal subgroup of $N(S)$, we have the quotient group $\mathcal{L}=N(\mathcal{S})/\langle\omega I\rangle\mathcal{S}$, whose elements represent all essentially different logical actions on $\mathcal{C}$. We now ask, does there exist a set of of logical operators $\overline{X}_i,\overline{Z}_i, i=1,2,\cdots,r$, satisfying the same commutation relation as any $r$ physical qudit Pauli operators $X_i,Z_i$ do , such that the logical Pauli group they generated is $\mathcal{L}$ by at most  phase factors, and thus $\mathcal{C}$ can be decomposed into $r$ logical qudits?
When $D$ is not prime, this may  not be the case, recall the example after Theorem  \ref{size of homo} which shows that even the dimension of $\mathcal{C}$ needs not always be $D^r$.  Recent advancements in the field have addressed this problem by utilizing  tools such as normal forms \cite{unknown} or by deduction from homological quantum rotor codes \cite{vuillot2023homological}. These developments provide deeper insights into the relationship between logical operators and the structure of code spaces.

   \section{Qudit Hypermap Code }
  In a general 2-complex construction, even if the 2-complex comes from an oriented surface, there seems to have no canonical way of orienting the edges, i.e, defining the functions \(I_s,I_t\). In this section, we show that this arbitrariness can be avoid when the 2-complex comes in a certain way from a hypermap.
   
   A hypermap, more precisely, a combinatorial hypermap, consists of a number set \(B_n=\{1,2,\cdots,n\}\) with a pair of permutations $(\alpha, \sigma)\in S_n$. For each element $\gamma\in\langle\alpha,\sigma\rangle$, define its orbits to be the equivalence classes of $B_n$ under the relation \(i\sim j\Leftrightarrow \exists \gamma'\in \langle\gamma\rangle, \gamma'(i)=j \), then for each $i\in B_n$, there is a positive integer \(r\) so that the \(\gamma\)-orbit it belongs to is \(orb_\gamma(i)=\{i,\gamma(i),\cdots,\gamma^{r-1}(i)\}\),with \(\gamma^r(i)=i\).  We call the orbits of \(\alpha\) \emph{hyperedges},  the orbits of \(\sigma\) \emph{hypervertices}, and the orbits of \(\alpha^{-1}\sigma\) \emph{faces}. For \(\alpha^{-1}\sigma\) here, we take the convention in Leslie,\cite{Martin} i.e, acting from left to right. In addition, we call the elements of $B_n$ themselves \emph{darts}, and  denote \(e_{ \owns i}\), \(v_{\owns i}\), and \(f_{\owns i}\) the hyperedge, the hypervertex, and the face that dart $i$ belongs to. 
    
    Let \(\mathcal{V},\mathcal{E},\mathcal{F}\) be the free $\mathbb{Z}_D$-modules generated by all hypervertices, hyperedges, and faces, respectively, $\mathcal{W}$ be the free $\mathbb{Z}_D$-modules generated by all darts $B_n$.  We define a homomorphism $d_2: \mathcal{F}\rightarrow\mathcal{W}$ by \(d_2(f)=\sum_{i \in f}i\), and a homomorphism $d_1: \mathcal{W}\rightarrow \mathcal{V}$ by \(d_1(i)=v_{\owns \alpha^{-1}(i)}-v_{\owns i}\) which takes the difference between the $\sigma$ orbits of two adjacent elements in an $\alpha$ orbit, then we have 
    \begin{lemma}
    \(d_1\circ d_2=0\).
    \end{lemma}
    \begin{proof}
    For an $f\in \mathcal{F}$, we write its element as $f=\{i_0,i_2,\cdots,i_{k-1}\}$ where the subscript $s\in\mathbb{Z}_k$ with $i_{s+1}=\alpha^{-1}\sigma(i_s)$, which implies \(v_{\owns \alpha^{-1}(i_s)}=v_{\owns i_{s+1}}\), thus 
    \(d_1\circ d_2(f)=d_1\sum_{s\in\mathbb{Z}_k}i_s=v_{\owns \alpha^{-1}(i_0)}-v_{\owns i_0}+v_{\owns \alpha^{-1}(i_1)}-v_{\owns i_1}+\cdots+v_{\owns \alpha^{-1}(i_{k-1})}-v_{\owns i_{k-1}}=0\).
    \end{proof}
    \noindent Also, there is a homomorphism $\iota:\mathcal{E}\rightarrow\mathcal{W}$ with \(\iota(e)=\sum_{i\in e}i\), which is very similar to $d_2$, and we have
    \begin{lemma}
    \( d_1\circ\iota=0\).
    \end{lemma}
    \noindent 
     Lemma 5 guarantees a well defined  homomorphism $\Delta_1$ from the quotient module $\mathcal{W}/\iota(\mathcal{E})$ to \(\mathcal{V}\), with $\Delta_1[\omega]=d_1\omega$, where $[\omega]$ denotes the equivalence class of $\omega$. Further more, if we define \(\Delta_2:\mathcal{F}\rightarrow \mathcal{W}/\iota(\mathcal{E})\) by $\Delta_2=\rho \circ d_2$, where $\rho$ is the natural projection
     from  \(\mathcal{W}\) to \(\mathcal{W}/\iota(\mathcal{E})\) , we would have $\Delta_1\circ \Delta_2=0$.
     \begin{figure}[h]
    
\centering
\begin{tikzpicture}[
squarednode/.style={rectangle, draw=black!0, fill=green!0, very thick, minimum size=2mm},
]
\node[squarednode]      (site 1)        {\(\mathcal{W}\)};
\node[squarednode]      (site 2)       [right=1.6 of site 1] {\(\mathcal{V}\)};
\node[squarednode]      (site 3)       [left=1.6 of site 1] {\(\mathcal{F}\)};
\node[squarednode]      (site 4)       [below=1 of site 1] {\(\mathcal{W}/\iota(\mathcal{E})\)};
\node[squarednode]      (fake l)       [left=0.3 of site 4]{};
\node[squarednode]      (fake r)       [right=0.3 of site 4]{};
\node[squarednode]      (p1)       [above=0.1 of fake l]{\(\Delta_2\)};
\draw[->] (site 3) -- (site 4);
\node[squarednode]      (p2)       [above=0.1 of fake r]{\(\Delta_1\)};
\draw[->] (site 4) -- (site 2);
\node[squarednode]      (fake l+)       [left=0.4 of site 1]{};
\node[squarednode]      (fake r+)       [right=0.4 of site 1]{};
\node[squarednode]      (d1)       [above=-0.2 of fake l+]{\(d_2\)};
\node[squarednode]      (d2)       [above=-0.2 of fake r+]{\(d_1\)};
\draw[->] (site 1) -- (site 2);
\draw[->] (site 3) -- (site 1);
\node[squarednode]      (center)       [below=0.3 of site 1]{};
\node[squarednode]      (p)       [right=-0.1 of center]{\(\rho\)};
\draw[->] (site 1) -- (site 4);
\end{tikzpicture}
\caption{\(\Delta_i\) are defined to make the diagram commute.}
\label{jiaohuantu}
\end{figure}

In the previous section, we discussed how to construct a homological quantum code in the qudit setting from a chain complex obtained from an abstract 2-complex. In fact, for any $\mathbb{Z}_D$-chain complex $C_2\stackrel{\partial_2}{\longrightarrow}C_1\stackrel{\partial_1}{\longrightarrow}C_0$, where $C_i$ are all free modules with finite bases ${f_i},{e_i},{v_i}$, if $\partial_2(f_j)=\sum_{i}\sigma_i e_i$ and $\delta_1(v^*_j)=\sum_i\sigma_i e^*_i$  is the unique expansion with the bases, where $\sigma_i\in \mathbb{Z}_D$, we can construct generators $B_{f_j}=\bigotimes_iZ^{\sigma_i}$ and $A_{v_j}=\bigotimes_iX^{\sigma_i}$. All relevant lemmas, theorems 
 along with their proofs in section 3 still apply, in particular, we have a homological quantum code $(\mathcal{C},\mathcal{S})$ that satisfies $\dim \mathcal{C}=|H_1|$. Now, if $\mathcal{W}/\iota(\mathcal{E})$ is a finitely generated free module, then we can directly construct a homological quantum code from the chain complex $\mathcal{F}\stackrel{\Delta_2}{\longrightarrow}\mathcal{W}/\iota(\mathcal{E})\stackrel{\Delta_1}{\longrightarrow}\mathcal{V}$. This can be seen as a generalization of the hypermap quantum code in the qudit case. To do this, we choose one dart from each hyperedge and call it a \emph{special dart}, and these special darts form a subset $S\subset B_n$. We have:

\begin{lemma}
\(\mathcal{W}/\iota(\mathcal{E})\) is a free module with a basis \(\{[i]\mid i\in B_n\setminus S\}\).
\end{lemma}
\begin{proof}
First, we show that this is a linear independent set. Suppose there are $k_i\in \mathbb{Z}_D$ such that $\sum_{i\in B_n\setminus S}k_i[i]=0$, then we have $\sum_{i\in B_n\setminus S}k_ii=\sum_eR_e\iota(e)$, $R_e\in \mathbb{Z}_D$. If we use $s_e$ to denote the special dart in the hyperedge $e$, then  \(\sum_eR_e\iota(e)=\sum_eR_es_e+\sum_{i\in B_n\setminus S}h_ii\) for some $h_i\in \mathbb{Z}_D$, thus $R_e=h_i-k_i=0$ by linear independence of the set $B_n$ in \(\mathcal{W}\), which further indicates $k_i=0$. 

On the other hand, for every $\omega\in\mathcal{W}$, we have some \(R_e,h_i\in\mathbb{Z}_D\) such that 
\begin{align*}
    [\omega]&=[\sum_eR_es_e+\sum_{i\in B_n\setminus S}h_ii]\\&=\sum_eR_e[s_e]+\sum_{i\in B_n\setminus S}h_i[i]\\&=\sum_eR_e(-{\sum_{i\in{e\setminus\{s_e\}}}[i]})+\sum_{i\in B_n\setminus S}h_i[i],
\end{align*}
 which shows \(\mathcal{W}/\iota(\mathcal{E})=\text{span}\{[i]\mid i\in B_n\setminus S\}\).
\end{proof}
\noindent Surprisingly, there is an another approach, where we will instead construct an abstract 2-Complex \(\Sigma=(V,E,I_s,I_t,F,B)\) whose \emph{chain} $C_2(\Sigma)\stackrel{\partial_2}{\longrightarrow}C_1(\Sigma)\stackrel{\partial_1}{\longrightarrow}C_0(\Sigma)$ is isomorphic to the chain of hypermap homology $\mathcal{F}\stackrel{\Delta_2}{\longrightarrow}\mathcal{W}/\iota(\mathcal{E})\stackrel{\Delta_1}{\longrightarrow}\mathcal{V}$, which shows that hypermap homology quantum codes are not something special, but simply the homological quantum codes from 2-Complexes.
 In order to define $\Sigma$, we let $V$ be the set of all hypervertices, $E$ be the set $B_n\setminus S$, and $F$ be the set of all faces, then It is apparent that we have \(C_2(\Sigma)= \mathcal{F} \), \(C_1(\Sigma)\simeq \mathcal{W}/\iota(\mathcal{E})\), and \(C_0(\Sigma)= \mathcal{V}\). Further more, for every $e\in E$, which is a non-special dart, i.e, $e=i\in B_n\setminus S $, define \(I_t(e)=v_{\owns \alpha^{-1}(i)}\), \(I_s(e)=v_{\owns i}\), then we have $\partial_1(e)=I_t(e)-I_s(e)=v_{\owns \alpha^{-1}(i)}-v_{\owns i}=d_1i=\Delta_1[i]$, which means $\partial_1\simeq\Delta_1$. The precise meaning of the symbol \(\simeq\) is that the matrix of \(\partial_1\) with respect to the basis \(E,V\) is the same as the matrix of \(\Delta_1\) with respect to the basis \([E],V\). To define $B$, notice that for every $f\in F$, there is a positive integer $r$ such that \(f=\{i_0,i_1,\cdots,i_{r-1}\}\) with  subscripts in $\mathbb{Z}_r$, and satisfy \(i_{k+1}=\alpha^{-1}\sigma (i_k)\) for all $k$. Suppose that the subset which consists of all special darts in $f$ is $S_f=\{i_{k_1},i_{k_2},\cdots i_{k_s}\}$, we have $[i_{k_t}]=[i_{k_t}-\iota(e_{\owns{i_{k_t}}})]=-\sum_{l=1}^{|e_{\owns{i_{k_t}}}|-1}[i^t_l]$, with $i^t_{l+1}=\alpha (i^t_{l})$ for all $l\in \{1,2,\cdots,|e_{\owns{i_{k_t}}}|-2 \}$, and also $\alpha(i_{k_t})=i^t_1$, $\alpha(i^t_{|e_{\owns{i_{k_t}}}|-1})=i_{k_t}$, where $i^t_l\in B_n\setminus S$. Thus we have 
 \begin{equation}
     \Delta_2(f)=\sum_{i\in f\setminus S}[i]-\sum_{t=1}^s{\sum_{l=1}^{|e_{\owns{i_{k_t}}}|-1}[i^t_l]}
 \end{equation} 
 along with the following lemma
 \begin{lemma}
 In the r-tuple \((i_0,i_1,\cdots,i_{r-1})\) from $f$, if we replace each $i_{k_t}\in S_f $ by the tuple $\mathbf{p}_t=((i_1^t)^{-1},(i_2^t)^{-1},\cdots,(i_{|e_{\owns{i_{k_t}}}|-1}^t)^{-1})$ in \(E\cup{E^{-1}}\), then we get a closed walk
 \begin{equation*}
     [i_0,i_1,\cdots,\mathbf{p}_1,i_{k_1+1},\cdots,\mathbf{p}_s,i_{k_s+1},\cdots,i_{r-1}].
 \end{equation*}
 \end{lemma}
 \begin{proof}
Suppose that  the `closed walk' above is re-indexed by $e_i\in E\cup{E^{-1}}$ with \(i \in \mathbb{Z}_K\) where $K$ is the length.   We need to verify that $I_s(e_{i+1})=I_t(e_i)$ for all $i$. For instance, let's consider the case where $e_{i+1} = (i^1_1)^{-1}$ and  \(i_{k_1-1}\notin S_f\), then we have $I_s((i^1_1)^{-1})=v_{\owns{\alpha^{-1}(i^1_1)}}=v_{\owns i_{k_1}}=v_{\owns \alpha^{-1}\sigma(i_{{k_1}-1})}=v_{\owns \alpha^{-1}(i_{{k_1}-1})}=I_t(i_{{k_1}-1})$. The other cases can be handled similarly.
 \end{proof}
 
\noindent Now, if we define \(B(f)\) to be the closed walk in lemma 7, then by equation (22) there is \(\partial_2\simeq\Delta_2\). 

We have shown that every hypermap map code is the homological quantum code constructed from a 2-Complex $\Sigma$. The most intresting observation about \(\Sigma\) is that it should be a surface 2-Complex. Actually, every hypermap \((\alpha,\sigma)\) has a geometrical representation $H=(M,\Gamma)$ called \emph{topological hypermap}, \cite{Martin,Zihan} where $M$ is an oriented surface, $\Gamma$ is an bipartite graph embedded in $M$ whose edges correspond to the darts in $B_n$,  the normal vector field given by $M$'s orientation determines the maps $\alpha$ and $\sigma$. Sarvepalli \cite{Pradeep} showed that we can obtain by adding curves on $M$ an ordinary surface code which equals the original hypermap code constructed from $H$ by Leslie. In our terms, Sarvepalli's curves corresponds to the set $E$ of edges in the 2-Complex $\Sigma$ of \((\alpha,\sigma)\) , furthermore,  they endow $M$ with a cell structure which equals combinatorically to \(\Sigma\). A small difference is that in Sarvepalli, \cite{Pradeep} the curves are not oriented for they only deal with qubit quantum codes. Notice that if a 2-Complex comes from a cell structure of an orientable surface with a global normal vector field, moreover, the  function $B:F\rightarrow W_\Gamma$ is determined by the global  field because the restriction of the  field in each 2-cell will induce an orientation of its boundary closed walk, then we must have  
\begin{equation}\label{O com}
    \sum_{f\in F} \partial_2(f)=0.
\end{equation}
\cite{Bombin} Let's define  orientable 2-Complexes to be those satisfying equation (\ref{O com}), then we have 
\begin{theorem}
The 2-Complex $\Sigma$ we constructed from \((\alpha,\sigma)\) is orientable.
\end{theorem}
\begin{proof}
We only have to show that \(\sum_{f\in F}\Delta_2(f)=0\), which is correct because \(\sum_{f\in F}d_2(f)=\sum_{i\in B_n}i=\sum_{e} \iota(e)\) where the last sum is done for all hyperedges. 
\end{proof}

\noindent However, apart from the fact that $\Sigma$ is orientable, we choose not to pursue the proof that $\Sigma$ is a Surface 2-Complex due to the potential intricacy of the argument.



 \bibliographystyle{unsrt}
\bibliography{name}

\appendix
\appendixpage
\addappheadtotoc

\section{Proof of Theorem 1}
We borrow the proof from \cite{standard}, with  emphasis on the need of excluding nontrivial scalar multiplications.
\begin{proof}
Define
\begin{equation}\label{def P}
P\mathrel{\mathop:}=\frac{1}{|\mathcal{S}|}\sum_{s\in \mathcal{S}}s.
\end{equation}
We first prove that $P$ is a projection operator on $\mathcal{C}$.

By definition, we have
\begin{equation}\label{Proj of qudit}
P=P^{\dag}=P^2,
\end{equation}
where the two equalities come from the group properties of $\mathcal{S}$. Equation (\ref{Proj of qudit}) tells us that $P$ is a projection operator \cite{LinA}.

Let $|\psi\rangle$ be an arbitrary vector in $\mathcal{C}$. Then, it follows from the definition in equation (\ref{def P}) that $P|\psi\rangle=|\psi\rangle$. Therefore, if we let $\mathcal{W}$ be the range of $P$, i.e., the space it projects onto, we have $\mathcal{C}\subset \mathcal{W}$. Now, let us choose an arbitrary $|\phi\rangle\in \mathcal{W}$. Then, we have
\begin{equation}\label{W no ele}
|\phi\rangle=P|\phi\rangle=\frac{1}{|\mathcal{S}|}\sum_{s\in \mathcal{S}}s|\phi\rangle.
\end{equation}
Multiplying both sides of equation (\ref{W no ele}) by an arbitrary $t\in \mathcal{S}$ and using the group properties of $\mathcal{S}$, we obtain
\begin{equation}
\begin{aligned}
t|\phi\rangle &=tP|\phi\rangle=\frac{1}{|\mathcal{S}|}\sum_{s\in \mathcal{S}}ts|\phi\rangle
=\frac{1}{|\mathcal{S}|}\sum_{s\in \mathcal{S}}s|\phi\rangle\\
&=P|\phi\rangle=|\phi\rangle.
\end{aligned}
\end{equation}
Since $t$ is arbitrary, we have
\begin{equation}
t|\phi\rangle=|\phi\rangle,\forall t\in \mathcal{S}.
\end{equation}
This means that $|\phi\rangle$ belongs to the stabilizer subspace $\mathcal{C}$, and hence $\mathcal{W}\subset \mathcal{C}$. Therefore, the range of $P$ is $\mathcal{C}$, and its trace is the dimension of its range, i.e.,
\begin{equation}
\text{Tr}(P)=K=\frac{1}{|\mathcal{S}|}D^n,
\end{equation}
where the second equality follows from the fact that the trace of any non-scalar element in $\mathcal{P}_n$ is zero, and the only scalar in $\mathcal{S}$ is $I$ whose trace is $D^n$.   This completes the proof.
\end{proof}

We remark here that Theorem \ref{size of sta} can be modified slightly such that $\mathcal{S}$ is not necessarily a subgroup of $\mathcal{P}_n$ without nontrivial scalar multiplication , but simply a finite subgroup of $U(\mathcal{H}_n)$ --- the set of unitary operators on $\mathcal{H}_n$, and the right hand side of equation (\ref{numberofk})  is replaced by \(\sum_{s\in \mathcal{S}} \text{Tr}(s)\). The corresponding part in the above proof is still valid.

\section{The proof of $|r(\mathcal{A})||r(\mathcal{A})^{\perp}|=D^{|E|}$ in Theorem 3. }

We restate the property as the following lemma whose proof is adapted from \cite{Free}, with some added and modified details.
\begin{lemma}
    Let \(D\geq 2\) be an integer and $n$ be a positive integer. For any submodule $E$ of $\mathbb{Z}^n_D$, let $E^{\perp}=\{x\in \mathbb{Z}^n_D\big|x\cdot y=0,\forall y\in E \}$, we have \(|E||E^{\perp}|=D^n\). 
\end{lemma}
\begin{proof}
For a $D$-th root of unity $\omega=e^{\frac{2\pi i}{D}}$, we first prove that
\begin{equation}\label{eqn:main}
\sum_{x\in E}\omega^{\eta\cdot x} =
\begin{cases}
      |E|, & \eta\in  E^{\perp}\\
      0 ,       & \eta\notin  E^{\perp}.
    \end{cases}
\end{equation}
When $\eta\in E^{\perp}$, equation (\ref{eqn:main}) is clearly true. Otherwise, we define a map $\eta: E\rightarrow \mathbb{Z}_D, x\mapsto \eta\cdot x$. Clearly, the map $\eta$ is a  module homomorphism, and since $\eta\notin E^{\perp}$, its image $\eta(E)$ is a nontrivial subgroup of $\mathbb{Z}_D$. We use the integers within 0 to $D-1$ as representatives for the elements of $\eta(E)$ and label them along with $D$ on the number axis.

  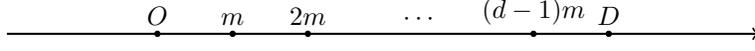
\begin{figure}[ht]
      \centering
      \begin{tikzpicture}
       \filldraw[white] (5.5,0) circle (1pt) node[anchor= south]{$\color{black}\cdots$};
           \draw[->,thick](0,0)--(10,0);
            \filldraw[black] (2,0) circle (1pt) node[anchor= south]{$O$};
             \filldraw[black] (8,0) circle (1pt) node[anchor= south]{$D$};
             \filldraw[black] (3,0) circle (1pt) node[anchor= south]{$m$};
              \filldraw[black] (4,0) circle (1pt) node[anchor= south]{$2m$};
               \filldraw[black] (7,0) circle (1pt) node[anchor= south]{$(d-1)m$};

      \end{tikzpicture}
      \caption{Representatives of elements in $\eta(E)$.}
      \label{Axis}
     \end{figure}

     Now, all the points in the number axis are equidistant from each other, otherwise there exist three consecutive points $a_1,a_2,a_3$ such that $a_3-a_2>a_2-a_1$ or $a_3-a_2<a_2-a_1$. Without loss of generality, assume the latter is true. Then, since $a_3-a_2\in \eta(E)$, we have $a_1+(a_3-a_2)\in \eta(E)$, but this contradicts the fact that $a_1<a_1+(a_3-a_2)<a_2$. We now let $m$ be the smallest positive integer labeled in the number axis and $|\eta(E)|=d$, then $D=md$ and $\eta(E)=\{0,m,2m,\cdots (d-1)m\}$, as shows in Figure \ref{Axis}. Therefore, we have
     \begin{equation}
      \begin{aligned}
          \sum_{k\in \eta(E) }\omega^k & =\sum^{d-1}_{l=0}\omega^{lm}\\
          & =1+\epsilon^1+\epsilon^2+\cdots+\epsilon^{d-1}\\
          & =\frac{1-\epsilon^d}{1-\epsilon}\\
          & =0.
      \end{aligned}
  \end{equation}
  Here, $\epsilon=\omega^m=e^{\frac{2\pi i}{D}\cdot \frac{D}{d}}=e^{\frac{2\pi i}{d}}$ represents the $d$-th root of unity, and since $\eta(E)$ is non-zero, we have $d>1$, implying $\epsilon\neq{1}$. By the fundamental theorem of module homomorphisms, $\eta$ induces an isomorphism $\overline{\eta}:E/\ker(\eta)\rightarrow\eta(E), x+\ker(\eta)\mapsto\eta(x)$. Therefore,
  \begin{equation}
     \begin{aligned}
         \sum_{x\in E}\omega^{\eta\cdot x} & =\sum_{k\in \eta(E)}\sum_{x\in\eta^{-1}(k)}\omega^{k}
          =|\ker(\eta)|\sum_{k\in \eta(E)}\omega^{k}\\
          & =0.
     \end{aligned}
 \end{equation}
 Now according to equation (\ref{eqn:main}) , we have:
 \begin{equation}
   \begin{aligned}
       |E||E^{\perp}|& =\sum_{\eta\in E^{\perp}}\sum_{x\in E}\omega^{\eta\cdot x}
        =\sum_{\eta\in \mathbb{Z}^n_D}\sum_{x\in E}\omega^{\eta\cdot x}\\
        & =\sum_{x\in E}\sum_{\eta\in \mathbb{Z}^n_D}\omega^{\eta\cdot x}
   \end{aligned}
  \end{equation}
  Let us define the submodule $E'=\mathbb{Z}_D^n$, then $E'^\perp={0}$. Using equation (\ref{eqn:main}) again, we have $\sum_{x\in E}\sum_{\eta\in \mathbb{Z}_D^n}\omega^{\eta\cdot x}=\sum_{\eta\in \mathbb{Z}_D^n}\omega^{\eta\cdot 0}=D^n$.
\end{proof}

 \end{document}